\documentclass[11pt]{article}

\usepackage{amsmath,amsthm,amssymb} %
\usepackage{color,graphicx} %
\usepackage{dsfont} %
\usepackage[hmargin=3.5cm,vmargin=2.5cm]{geometry} %
\usepackage[dvipdfm]{hyperref} %
\usepackage{suffix} % for *-version commands
\usepackage{times} %
\usepackage{xspace} %

\hypersetup{colorlinks=true, linkcolor=blue, citecolor=magenta}

\newcommand\ket[1]{\ensuremath{|#1\rangle}} %
\newcommand\bra[1]{\ensuremath{\langle#1|}} %
\newcommand\tr{\mathop{\rm tr}\nolimits} %
\newcommand\class[1]{{\bf #1}} %
\newcommand\problem[1]{{\tt #1}} %
\WithSuffix\newcommand\problem*[1]{\problem{#1}$^*$} %
\newcommand\yes{\textsf{yes}} %
\newcommand\no{\textsf{no}} %
\newcommand\true{\texttt{true}\xspace} %
\newcommand\false{\texttt{false}\xspace} %

\def\1{I} %

\def\A{\mathcal{A}} %
\def\B{\mathcal{B}} %
\def\H{\mathcal{H}} %
\def\I{\mathcal{I}} %
\def\P{\mathcal{P}} %

\newtheorem{lemma}{Lemma} %
\newtheorem{theorem}{Theorem} %
\begin{document}

%% End-Of-Header

\title{\bf\Large Binary Constraint System Games and Locally
  Commutative Reductions}

\author{Zhengfeng Ji\\[.2em]
  \it\normalsize Institute for Quantum Computing and School of Computer Science,\\
  \it\normalsize University of Waterloo, Waterloo, Ontario, Canada\\[.1em]
  \it\normalsize State Key Laboratory of Computer Science, Institute of Software,\\
  \it\normalsize Chinese Academy of Sciences, Beijing, China}

\date{}

\maketitle

\begin{abstract}
  A binary constraint system game is a two-player one-round non-local
  game defined by a system of Boolean constraints. The game has a
  perfect quantum strategy if and only if the constraint system has a
  quantum satisfying assignment~\cite{CM12}. We show that several
  concepts including the quantum chromatic number and the
  Kochen-Specker sets that arose from different contexts fit naturally
  in the binary constraint system framework. The structure and
  complexity of the quantum satisfiability problems for these
  constraint systems are investigated. Combined with a new construct
  called the commutativity gadget for each problem, several classic
  \class{NP}-hardness reductions are lifted to their corresponding
  quantum versions. We also provide a simple parity constraint game
  that requires $\Omega(\sqrt{n})$ EPR pairs in perfect strategies
  where $n$ is the number of variables in the constraint system.
\end{abstract}

\section{Introduction}\label{sec:intro}

Satisfiability, the problem of deciding if there is a truth assignment
to a given Boolean formula, is arguably one of the most important
problems in computer science. It is the first known example of an
\class{NP}-complete problem~\cite{Coo71}. The problem itself and some
of its variants such as \problem{3-SAT} and \problem{NAE-SAT} are the
usual starting points of \class{NP}-hardness proofs (see
e.g.~\cite{Pap93}).

A quantum version of the satisfiability problem arose recently from
the study of quantum two-player one-round games~\cite{CM12}. It is a
concept that has roots in both the classical complexity
theory~\cite{BFL91,BGKW88,FL92} and the study of entanglement and
non-locality of quantum systems~\cite{Bel64,KS67,Tsi87,Mer90}. The
topic of this paper is to explore the structure and relations of
problems under this new definition of satisfiability. As in the
classical theory, polynomial-time reductions will be the major
methodology employed in our investigations.

Two-player one-round games have rich applications in classical
complexity theory~\cite{BGKW88,BFL91,FL92,Raz98,Has01}. A two-player
one-round game consists of two cooperating players Alice, Bob and a
referee. Alice and Bob can agree on a strategy beforehand but cannot
communicate after the game starts. The referee randomly chooses a pair
of questions $(s, t)$ for Alice and Bob according to some fixed
probability distribution $p$ over $S\times T$, and then sends $s$ to
Alice and $t$ to Bob. Alice and Bob are required to reply with answers
$a, b$ respectively to the referee from the finite answer sets $A, B$.
Finally, the referee determines if they win the game according to a
predicate $V:A\times B\times S\times T\rightarrow \{0, 1\}$. Let the
classical value of the game be the optimal winning probability for
Alice and Bob.

In the quantum case, the players Alice and Bob are allowed to share a
quantum state before the game starts and can perform quantum
measurements to obtain the answers. Such two-player games are
sometimes called non-local games~\cite{CHTW04} in the literature. Let
the shared state be $\ket{\psi}$ and let Alice and Bob's measurements
be $\{ A_s^a \}$ and $\{ B_t^b \}$ respectively depending on the
questions $s, t$ they receive. The optimal winning probability for
Alice and Bob, also known as the non-local value of the game, is
therefore
\begin{equation}\label{eq:value}
  \sup\, \biggl\{ \sum_{a,b,s,t} p(s,t) V(a,b \mid s,t)
  \bra{\psi} A_s^a \otimes B_t^b \ket{\psi} \biggr\},
\end{equation}
where the supremum is taken over all possible Hilbert spaces $\H_A$
and $\H_B$, all quantum state $\ket{\psi} \in \H_A\otimes \H_B$, and
all positive-operator valued measures (POVMs) $\{A_s^a\}$ and
$\{B_t^b\}$ such that $\sum_aA_s^a =\1$, $\sum_bB_t^b=\1$. Notions
such as quantum states and measurements used in this paper are
standard in quantum information theory and we refer readers not
familiar with them to~\cite{NC00}. The non-local value of a game may
be strictly larger than its classical value. For instance, the CHSH
game, a simple example that recasts the CHSH inequality~\cite{CHSH69}
in the non-local game framework, has non-local value approximately
$0.85$ and classical value $0.75$. The investigation of the properties
of non-local games and the complexity of determining the non-local
value is an important research
topic~\cite{CHTW04,CSUU08,KRT10,KV11,IV12,Vid13}.

Binary constraint system games constitute a special class of non-local
games defined in~\cite{CM12}. It relates a binary constraint system to
a non-local game in a natural way. A binary constraint system (BCS) is
a collection of Boolean constraints $C_1, C_2, \ldots, C_m$ over
binary variables $x_1, x_2, \ldots, x_n \in \{0, 1\}$. For example,
Eq.~\eqref{eq:magic-square} below defines a BCS of six constraints. A
BCS is classically satisfiable if there exists a truth assignment to
the variables that satisfies all constraints. A BCS defines a
non-local game in the following way. The referee picks a constraint
$C_s$ and a variable $x_t$ occurring in $C_s$ uniformly at random,
sends $s,t$ to Alice and Bob. Alice is required to give binary
assignments to each of the variables in $C_s$ and Bob is required to
give a binary assignment to $x_t$. They win the game if Alice's
assignment does satisfy the constraint $C_s$ and Bob's assignment to
$x_t$ is equal to Alice's assignment to the corresponding variable. We
call the non-local game thus defined a BCS game. A parity BCS is a
constraint system whose constraints are of the form
\begin{equation*}
  x_{i_1} \oplus x_{i_2} \oplus \cdots \oplus x_{i_k} = 0 \text{ (or }
  1\text{)}.
\end{equation*}
The corresponding game will be called a parity BCS game.

The magic square game~\cite{Mer90,Per90,Mer93,Ara04} is the
paradigmatic example of a parity BCS game. In this game, we have nine
variables $x_1, x_2, \ldots, x_9 \in \{0, 1\}$ and the following six
constraints
\begin{equation}\label{eq:magic-square}
  \begin{split}
    x_1 \oplus x_2 \oplus x_3 = 0, \qquad x_1 \oplus x_4 \oplus x_7 = 0,\\
    x_4 \oplus x_5 \oplus x_6 = 0, \qquad x_2 \oplus x_5 \oplus x_8 = 0,\\
    x_7 \oplus x_8 \oplus x_9 = 0, \qquad x_3 \oplus x_6 \oplus x_9 =
    1.
  \end{split}
\end{equation}
It gets the name because we can arrange the variables in the following
$3$ by $3$ square in Fig.~\ref{fig:magic-square}
\begin{figure}[htbp]
  \centering
  \includegraphics{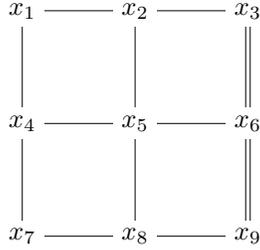}
  \caption{Magic square game}
  \label{fig:magic-square}
\end{figure}
and the constraints are that the rows have even parity, the first two
columns have even parity and the last column has odd parity. In the
figure, we use single and double lines to represent even and odd
parity constraints respectively.

It is easy to see that the BCS of the magic square game in
Eq.~\eqref{eq:magic-square} is unsatisfiable classically.
Correspondingly, if Alice and Bob cannot share entanglement and use
quantum strategies, they will not be able to win the magic square game
with probability $1$ without communication. Surprisingly, however,
they can win the game perfectly with the help of entanglement in the
non-local game setting~\cite{Mer90,Mer93}. In fact, if we define a
shared EPR pair as the two-qubit state $(\ket{00}+\ket{11}) /
\sqrt{2}$ shared between Alice and Bob, then two shared EPR pairs
suffice to win this particular game perfectly. We will call a strategy
that wins a game with probability $1$ a perfect strategy. Another
interesting BCS game that has a perfect strategy is the so-called
magic pentagram game (see~\cite{Mer93,CM12,Ark12} for further
discussions on this game). Generally, a magic BCS game is a BCS game
that has a perfect quantum strategy but no classical strategy.

Recently, R. Cleve and R. Mittal gave a characterization of BCS games
that have perfect strategies~\cite{CM12}. This characterization
establishes the concept of quantum satisfiability for binary
constraint systems. It relates the BCS game that has a perfect
strategy on one hand and the BCS that has a quantum satisfying
assignment on the other. Because of the importance of the definition
to this work, we will discuss it in detail in
Section~\ref{sec:quantum_assignment}.

Currently, we do not know much about BCS games. For example, it is not
known if it is decidable to determine the existence of a perfect
strategy and it is not known whether there is an upper bound on the
amount of entanglement needed in perfect strategies. One exception is
that, for parity BCS games where each variable occurs at most twice in
all the constraints, a beautiful criterion that decides if a perfect
strategy exists was found in~\cite{Ark12}. This was done by
considering the graph planarity of the so-called intersection graph of
the game and an application of the Pontryagin-Kuratowski theorem in
graph theory. It is further shown in that paper that when a perfect
strategy exists, three shared EPR pairs suffice in all parity BCS
games of this kind.

Our contribution to this subject includes three different aspects.
First, we observe that many previously known concepts such as the
quantum chromatic number~\cite{CMN+07} and the Kochen-Specker
sets~\cite{KS67} can actually be described in the BCS framework, thus
enriching the scope of this subject. Second, we initiated the study of
reductions between binary constraint systems that preserves quantum
satisfiability. These reductions are usually built up on their
classical counterpart, but require non-trivial modifications by a new
construct called commutativity gadget. For example we will be able to
reduce the quantum analog of \problem{3-SAT} to that of
\problem{3-COLORING}. The non-commutative Gr\"{o}bner basis
computation is employed to assist our design of commutative gadgets.
Third, an example of a parity BCS game is given where a large amount
of entanglement is necessary for perfect strategies to exist.
Moreover, we can have such games where each variable occurs in at most
three different constraints. This example is interesting when compared
with the results obtained in~\cite{Ark12}. The core of the
construction is the Clifford algebra and its representation theory,
which have already found applications in the study of non-local
games~\cite{Tsi87,Slo11}. Our result is yet another use of the
Clifford algebra in a different way.

\section{Quantum Satisfying Assignments}

\label{sec:quantum_assignment}

In this section, we review the definition of a quantum satisfying
assignment of a binary constraint system introduced in~\cite{CM12}.
Before stating the definition, we need to first rewrite the Boolean
constraints as polynomial constraints. For parity games, it is
convenient to work with the $\{\pm 1\}$ domain so that exclusive OR
operation can be replace by multiplication, while for most of the
other cases, the $\{0, 1\}$ domain turns out to be more convenient.
The choice of domain actually does not make much difference in the
problem.

For example, the magic square BCS in Eq.~\eqref{eq:magic-square} can
be written as
\begin{equation}
  \label{eq:magic-square-prod}
  \begin{split}
    & x_1 x_2 x_3 = 1, \qquad x_1 x_4 x_7 = 1,\\
    & x_4 x_5 x_6 = 1, \qquad x_2 x_5 x_8 = 1,\\
    & x_7 x_8 x_9 = 1, \qquad x_3 x_6 x_9 =-1,\\
    & x_j^2=1 \text{ for } j=1,2,\ldots,9.
  \end{split}
\end{equation}
One can find more examples of BCS in Section~\ref{sec:examples}.

A BCS has a {\it quantum satisfying assignment\/} or an {\it operator
  assignment\/}, if there exists a {\it finite dimensional\/} Hilbert
space $\H$, and an assignment of a self-adjoint linear operator $X_j
\in \B(\H)$ to each of its variables $x_j$ such that the following
conditions hold
\begin{itemize}
\item[(a)] The operators satisfy each polynomial constraint when we
  substitute $x_j$'s with operators $X_j$'s respectively and $1$ with
  $\1$,
\item[(b)] For each $j$, the spectrum of $X_j$ is contained in $\{ 0,
  1\}$ ($\{ \pm 1 \}$), or equivalently, $X_j$ is self-adjoint and
  $X_j^2=X_j$ ($X_j^2=\1$),
\item[(c)] Each pair of operators $X_j$, $X_k$ that appear in the same
  constraint is commuting, $X_jX_k=X_kX_j$.
\end{itemize}
Conditions (a) and (b) naturally follow from the defining constraints
of the BCS. Condition (c) is implicit in the classical case, and is
made explicit in the quantum case so that measurements used on Alice's
side to determine assignments to the variables are compatible. Another
nice property guaranteed by (c) is that each polynomial constraint
$C_j$ will evaluate to self-adjoint operator if the variables are all
self-adjoint. We will refer to this condition as the local
commutativity condition of a BCS in the following.

For the magic square constraint system, the nine operators in
Fig.~\ref{fig:magic-square-solution} constitute a quantum satisfying
assignment on a $4$-dimensional Hilbert space which satisfies all
constraints in Eq.~\eqref{eq:magic-square-prod}. In the figure, $I$ is
the identity matrix of size $2$ and $X,Y,Z$ are Pauli matrices
\begin{equation*}
  X =
  \begin{pmatrix}
    0 & 1\\
    1 & 0
  \end{pmatrix}, Y =
  \begin{pmatrix}
    0 & -i\\
    i & \phantom{-}0
  \end{pmatrix}, Z =
  \begin{pmatrix}
    1 & \phantom{-}0\\
    0 & -1
  \end{pmatrix}.
\end{equation*}

\begin{figure}[htbp]
  \centering
  \includegraphics{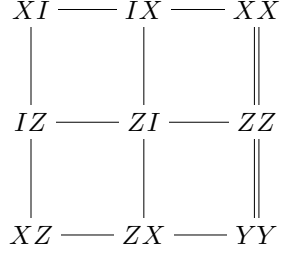}
  \caption{Quantum satisfying assignment for the magic square game}
  \label{fig:magic-square-solution}
\end{figure}

One of the main results in~\cite{CM12} states that a BCS game has a
perfect strategy if and only if the corresponding BCS has a quantum
satisfying assignment. In the original statement of the theorem, it is
required that the perfect strategy uses finite or countably infinite
dimensional entanglement. We note, however, that as long as a perfect
strategy exists in an arbitrary Hilbert space $\H$, the state
$\ket{\psi} = \sum \alpha_j\ket{\phi_j}\ket{\psi_j}$ used in the
strategy always has at most countably infinite many nonzero amplitudes
$\alpha_j$~\cite[Subsection 2.2]{KR97}. Therefore, the theorem still
holds without the dimension assumption of the shared state. During the
proof of this result, it was also made clear that one could always use
maximally entangled states (states of the form $\sum_j\ket{j,j}$) in a
perfect strategy for BCS games. This is a property of BCS games that
does not hold in general non-local games~\cite{Reg11}.

A BCS problem, or $*$-problem, is the decision problem that asks
whether a quantum satisfying assignment exists for a given encoding of
BCS as input. For particular constraint systems, we will denote the
corresponding BCS problem by adding a star to the name of the
corresponding classical problem. For example, \problem*{3-SAT} denotes
the BCS problem defined by \problem{3-SAT} (the satisfiability problem
where each constraint is the disjunction of at most three literals)
instances. It is important to clarify that the quantum analogs of
satisfiability problems considered here are different from the local
Hamiltonian related problems~\cite{KSV02,AN02,Bra06,GN13} in
Hamiltonian complexity theory including, for example, the
\problem{Quantum-3-SAT} problem recently shown to be
\class{QMA$\bf_1$}-complete~\cite{GN13}.

One may also consider the decision problem that asks whether a BCS
game has non-local value $1$. It is not known if this problem is the
same as the BCS problem defined above. The main problem is that, by
the definition in Eq.~\eqref{eq:value}, the non-local value could be
$1$ for games without perfect strategy but have a sequence of
non-perfect strategies whose values converge to $1$.

The transformation of the existence of a perfect strategy to that of
an operator assignment of a BCS~\cite{CM12} opens the door to studying
the problem from an algebraic point of view. The main problem is to
tell if a set of ``non-commutative'' polynomials has self-adjoint
operator solutions, and can be considered in some sense as the matrix
version of the problem of deciding whether a set of polynomials has
real solutions.

We give the definitions of several algebraic concepts to more formally
convey the idea. Let $F$ denote the free semigroup of $n$ generators
$x_1, x_2, \ldots, x_n$, and let $\A$ be its semigroup algebra
$\mathbb{C}[F]$. Elements of $F$ will be called words, and the empty
word is denoted by $1$. Each element $p$ of algebra $\A$ is a formal
finite sum of words in $F$, $p = \sum \alpha_w w$ where $\alpha_w \in
\mathbb{C}$. In the following, elements of $\A$ will be referred to as
non-commutative polynomials or simply polynomials in short. Equip $\A$
with an involution~$*$ defined as
\begin{equation*}
  w^* = x_{j_k} x_{j_{k-1}} \cdots x_{j_1} \text{ for }
  w = x_{j_1}x_{j_2}\cdots x_{j_k} \in F,
\end{equation*}
and $p^* = \sum \overline{\alpha_w} w^*$ for $p = \sum \alpha_w w \in
\A$ where $\overline{\alpha_w}$ is the complex conjugate of
$\alpha_w$. A polynomial $p \in \A$ is called self-adjoint if $p^*=p$.
For example, $x_1x_2-x_2x_1$ is not self-adjoint but
$i(x_1x_2-x_2x_1)$ is where $i=\sqrt{-1}$. Denote the space of all
self-adjoint polynomials as $\A_S$. One can evaluate a polynomial $p$
at the tuple of self-adjoint bounded linear operators $X=(X_1, X_2,
\ldots, X_n)$, written as $p(X)$, as the image of the $*$-algebra
homomorphism that maps $x_j$ to $X_j$ for each $j$. It is easy to see
if $p$ and the $X_j$'s are all self-adjoint, so is $p(X)$.

Any BCS of variables $x_1, x_2, \ldots, x_n$ and constraints $C_1,
C_2, \ldots, C_m$ can be equivalently described by a set $\P$ of
self-adjoint polynomials, which includes the three types of
polynomials of the form $C_j + C_j^*$, $1-x_j^2$ and also
$i(x_kx_j-x_jx_k)$ where $x_j, x_k$ occur in the same constraint, each
of which corresponds to one of the three conditions (a), (b) and (c)
for an operator assignment of the BCS.

Let $\B$ be a constraint system defined by a set $\P$ of self-adjoint
non-commutative polynomials. Denote the ideal $\I$ of $\A$ generated
by $\P$ as the ideal of $\B$, and the quotient $\A/\I$ as the algebra
of $\B$. The algebra of a constraint system is associative, but
usually non-commutative and is hard to analyze. In the analysis of
locally commutative reductions, we usually need to know if the
commutator of two generators is in the ideal $\I$. Non-commutative
Gr\"{o}bner basis theory~\cite{Mor94,Gre99} comes into play in these
type of investigations and turns out to be successful. Some of our
proofs of commutativity gadgets are indeed first obtained by
non-commutative Gr\"{o}bner basis computations, using packages such as
{\it GBNP\/}~\cite{CK10} and {\it NCGB\/}~\cite{HS97}.

\section{Examples of BCS games}\label{sec:examples}

One of the major motivation that leads to the definition of a BCS game
is the magic square game and its closely related variants such as the
pentagram game. The previously known magic BCS games are isolated
examples usually employing only parity constraints. In this section,
we will show that the BCS game is actually a versatile framework, and
that most of the known non-local games with perfect strategy are
usually BCS games in disguise. For example, the quantum chromatic
number~\cite{CMN+07}, the Kochen-Specker sets~\cite{KS67}, the
Deutsch-Jozsa game~\cite{BCT99}, etc., can either be completely
described in the BCS framework or give non-trivial magic BCS games at
least. More generally, any non-local game that has a perfect strategy
using projective measurements on a maximally entangled state gives
rise to a BCS.

\subsection{BCS for satisfiability problems}

We start the discussion by considering the most natural examples of
binary constraint systems, those defined by a Boolean formula. In a
Boolean formula, a literal is either a variable (positive literal) or
a negation of a variable (negative literal). A clause is a disjunction
of literals. Most of the time, we will be working with formulas of the
conjunctive normal form, which is a conjunction of clauses. For a
formula of the form $\bigwedge_{j=1}^m C_j$, the BCS associated with
it is the constraint system of $\{ C_j \}_{j=1}^m$.

In this paper, we will discuss several $*$-versions of the
satisfiability problems including \problem{$k$-SAT} where each clause
refers to at most $k$ literals, \problem{1-in-3-SAT} where a clause is
satisfied if it refers to at most three literals and exactly one of
them is assigned \true, and \problem{HORN-SAT} where each clause has
at most one positive literal.

\subsection{BCS for quantum chromatic number}

\label{sec:quantum_chromatic_number}

The quantum chromatic number of a graph $G$~\cite{CMN+07,SS12,MSS13},
denoted by $\chi^*(G)$, is the minimum number of colors necessary for
a non-local game defined by $G$ to have a perfect
strategy~\cite{CMN+07}. In the game, Alice and Bob each receive a
vertex $u, v$ respectively, and are required to send back to the
verifier the color $\alpha, \beta \in \{0, \ldots, k-1\}$ of the
vertex they got. Two conditions have to be met for them to win the
game: (1) if $u=v$, then $\alpha=\beta$, and (2) if $(u,v) \in E(G)$,
then $\alpha \ne \beta$. If Alice and Bob are restricted to classical
strategies only, the minimum number $k$ necessary for them to win the
game is exactly the (classical) chromatic number of the graph
$\chi(G)$. The quantum chromatic number $\chi^*(G)$ is defined to be
the smallest number $k$ such that Alice, Bob can win the game with a
quantum strategy. The quantum $k$-coloring problem, denoted as
\problem*{$k$-COLORING}, is the decision problem that asks if
$\chi^*(G) \le k$ given $G$ as input.

A constraint system arises naturally for the quantum coloring problem.
To each vertex $v$ in graph $G$, assign $k$ binary variables $x_{v,0},
x_{v,1}, \ldots, x_{v,k-1} \in \{ 0, 1 \}$. These variables are
indicators of whether $v$ has color $\alpha$, and therefore sum to
$1$. For $(u, v) \in E(G)$, we require that $x_{v,\alpha}x_{w,\alpha}
= 0$ and $x_{w,\alpha}x_{v,\alpha} = 0$. The constraint system is
therefore
\begin{subequations}\label{eq:bcs_color}
  \begin{align}
    x_{v,0} + x_{v,1} + \cdots + x_{v,k-1} = 1, & \text{ for all
      vertices } v\in V(G),\label{eq:bcs_color_1}\\
    x_{v,\alpha}x_{w,\alpha} = 0, & \text{ for all adjacent } v,w,
    \text{ and } \alpha= 0, 1, \ldots, k-1,\label{eq:bcs_color_pair}\\
    x_{v,\alpha}^2 - x_{v,\alpha} = 0, & \text{ for all } v\in V(G),
    \text{ and } \alpha = 0, 1, \ldots, k-1.\label{eq:bcs_color_2}
  \end{align}
\end{subequations}

It is easy to see that $\chi(G) \le k$ if and only if the above
constraint system has a real (and therefore $0, 1$) solution. If the
quantum chromatic number $\chi^*(G) \le k$, then the analysis
in~\cite{CMN+07} (Proposition 1, Eq.~(4) and the discussion above the
equation) guarantees that the measurement operators of Alice is an
operator assignment of the constraint system Eq.~\eqref{eq:bcs_color}.
In fact, the Eqs.~\eqref{eq:bcs_color_1} and~\eqref{eq:bcs_color_2}
correspond to the fact that Alice uses projective measurements, and
Eq.~\eqref{eq:bcs_color_pair} corresponds to the Eq.~(4)
of~\cite{CMN+07} (namely, $E_{v\alpha} E_{w\alpha} = 0$). On the other
hand, if the constraint system in Eq.~\eqref{eq:bcs_color} has an
operator assignment, it follows easily that $\chi^*(G) \le k$ again
by~\cite{CMN+07}. Therefore, the BCS problem defined by this
constraint system is a \yes-instance, if and only if the graph $G$ has
quantum chromatic number $\chi^*(G) \le k$. By the above discussions,
we can study the quantum chromatic number and the
\problem*{$k$-COLORING} problem in the BCS framework, although the
derived BCS game will be slightly different from the coloring game
considered in~\cite{CMN+07}. We will call an operator assignment of
this BCS a $k$-coloring operator assignment of $G$.

We have shown how one can represent the quantum graph coloring problem
by a BCS. Conversely, it is also possible to represent any BCS problem
as a generalized graph coloring problem, called the constraint graph
coloring problem. A constraint graph coloring problem is defined by a
graph $G=(V,E)$, a set of colors $\Sigma_v$ allowed for each vertex
$v\in V$ and a constraint $c_e$ for each edge $e=(u,v)\in E$ where
$c_e: \Sigma_u \times \Sigma_v \rightarrow \{0, 1\}$. The ordinary
graph coloring problem is a constraint graph problem where
$\Sigma_v$'s are the same and $c_e$ is defined such that $c_e(\alpha,
\beta) = 1$ if and only if $\alpha \ne \beta$. One can define the
quantum colorability of a constraint graph by a two-player game
similar to that for the quantum coloring problem~\cite{CMN+07}. The
constraint graph problem defined by a BCS $\{ C_j \}$ is on a
bipartite graph that has a vertex for each constraint $C_j$ and a
vertex for each variable. The set of colors for the constraint vertex
$C_j$ is the set of all possible assignment to the variables in the
constraint, and the set of colors for the variable vertex is always
$\{0, 1\}$. If a variable $x$ occurs in a constraint $C_j$, then there
is an edge between them whose constraint checks if the coloring of
$C_j$ is a satisfying assignment of the constraint and if the
assignment to $x$ is consistent with that in $C_j$. The game defined
by the constraint graph coloring problem associated to a BCS is
slightly different from the BCS game defined in~\cite{CM12}. The
difference comes from the consistency checks in the constraint graph
coloring game, which turn out to be redundant for constraint graphs
defined by BCS games as indicated by~\cite{CM12}.

\subsection{BCS for Kochen-Specker sets}

Another important concept that fits well in the BCS framework is the
Kochen-Specker sets~\cite{KS67,Hel13}. Generally, a Kochen-Specker set
is a set of projections $S = \{ P_j \}$ such that there is no $0,
1$-valued function $h:S \rightarrow \{ 0, 1\}$ satisfying the
condition: $\sum_{P_j \in B} h(P_j) = 1$ for any subset $B$ of $S$
such that $\sum_{P_j\in B} P_j = I$. Most of the examples of
Kochen-Specker sets consist solely of rank-one projections, in which
case the set can also be described by a set of unit vectors. A set $S
= \{ u_i \}$ of unit vectors in $\H$ is a Kochen-Specker set if there
is no function $h:S\rightarrow \{ 0, 1 \}$ such that for any subset
$B$ of $S$ forming an orthonormal basis of $\H$, $\sum_{u\in B}h(u) =
1$. The first finite construction of Kochen-Specker set consists of
$117$ vectors in $\mathbb{R}^3$~\cite{KS67}, but this number has been
reduced to $31$ by Conway and Kochen~\cite{Per02}.

To describe Kochen-Specker sets in the BCS framework, let us consider
the following linear constraint system of a set $S$ of binary
variables $x_j$
\begin{subequations}\label{eq:bcs_ks}
  \begin{align}
    \sum_{x_j\in B_k} x_j & = 1, \text{ for some } B_1, B_2, \ldots,
    B_m \subset S, \label{eq:bcs_ks_1}\\
    x_j^2 - x_j & = 0, \text{ for all } x_j \in S.
  \end{align}
\end{subequations}
The classical assignment of the constraint system corresponds exactly
to the $0, 1$-valued function $h$ in the definition of Kochen-Specker
sets. Therefore, for any BCS in Eq.~\eqref{eq:bcs_ks} that has a
quantum assignment but no classical assignment, the quantum assignment
for the BCS forms a Kochen-Specker set. On the other hand, any
Kochen-Specker set of projections defines a magic constraint system of
the above form where $B_j$'s are chosen to be all the subsets of
projections summing to $I$. We briefly mention that a weak variant of
Kochen-Specker sets is also investigated in the literature~\cite{RW04}
and it is easy to see that these sets also correspond to BCSs by using
constraints of the form $x_jx_k = 0$. Relations between weak
Kochen-Specker sets and certain pseudo-telepathy games are known
previously~\cite{RW04,MSS13}.

For later references, we denote \problem*{KOCHEN-SPECKER} to be the
$*$-problem defined by binary constraint systems whose constraints are
all of the form $\sum_{x_j\in B} x_j = 1$ as in
Eq.~\eqref{eq:bcs_ks_1}.

\subsection{BCS from general non-local games}

More generally, any non-local game defines a BCS such that the game
has a perfect strategy using maximally entangled states and projective
measurements if and only if the corresponding BCS has quantum
satisfying assignment. This is, in some sense, a weak converse of the
main theorem in~\cite{CM12}, as any BCS game with a perfect strategy
uses maximally entangled states and projective measurements.

For any non-local game with distribution $p$ on questions $S\times T$
and verifier $V:A\times B\times S\times T\rightarrow \{0, 1\}$, define
a BCS of the game in the following
\begin{subequations}\label{eq:bcs_nonlocal}
  \begin{align}
    \sum_{a\in A} x_{s,a} & = 1, \text{ for all } s\in S,\\
    \sum_{b\in B} y_{t,b} & = 1, \text{ for all } t\in T,\\
    x_{s,a}y_{t,b} & = 0, \text{ for } p(s,t) > 0, V(a,b,s,t) = 0,\\
    x_{s,a}^2-x_{s,a} & = 0, \text{ for all } s\in S, a\in A,\\
    y_{t,b}^2-y_{t,b} & = 0, \text{ for all } t\in T, b\in B.
  \end{align}
\end{subequations}

\begin{lemma}\label{lem:game2bcs}
  A non-local game has a perfect strategy using maximally entangled
  state and projective measurements if and only if the corresponding
  BCS has a quantum satisfying assignment.
\end{lemma}

\begin{proof}
  Let ${A_s^a}$, ${B_t^b}$ be the projective measurements of Alice and
  Bob, and let $\ket{\phi}$ be the maximally entangled state they use.
  They win the non-local game with probability $1$ if and only if they
  never give ``forbidden'' answers $a, b$ when $s, t$ are asked and
  $V(a,b,s,t) = 0$. This means that $\bra{\phi} A_s^a\otimes B_t^b
  \ket{\phi} = 0$, which simplifies to $\tr(A_s^a\bar{B}_t^b) = 0$ for
  such $(a,b,s,t)$'s. As both $A_s^a$ and $\bar{B}_t^b$ are positive
  semidefinite, it follows that $A_s^a \bar{B}_t^b = 0$. Therefore,
  $x_{s,a} = A_s^a$ and $y_{t,b} = \bar{B}_t^b$ is an operator
  assignment of the BCS in Eq.~\eqref{eq:bcs_nonlocal}. Conversely, if
  the BCS has an operator assignment, then the operator assignment
  gives a perfect strategy for Alice and Bob.
\end{proof}

Lemma~\ref{lem:game2bcs} can be used to construct different magic BCS
instances from many preexisting examples of pseudo-telepathy
games~\cite{BBT05}, as long as the game is a two player game that uses
maximally entangled state and projective measurements. The
Deutsch-Jozsa games~\cite{BCT99}, the so called matching
games~\cite{BJK04,BBT05}, etc., are such examples that can give
non-trivial magic binary constraint systems.

\subsection{BCS for non-binary constraint systems}

Suppose we have a constraint system where the domain of the variables
is not binary. Can we define quantum satisfiability for such
constraint systems? One simple solution is to introduce indicator
variables, as hinted by the quantum $3$-coloring case discussed in
Subsection~\ref{sec:quantum_chromatic_number}. For example, if the
domain of a variable $x$ is $\{0, 1, \ldots, k-1\}$, one can use $k$
variables $x_\alpha$ for $\alpha = 0, 1, \ldots, k-1$ each indicating
whether $x=\alpha$. This way, we can rewrite the constraints in terms
of the new indicator variables and the problem is translated to the
binary case.

One caveat is that different choices of classically equivalent initial
constraint system may result in BCS systems that are not equivalent in
the quantum sense. Take the graph coloring problem as an example. Let
$x_v$ be the variable representing the color of the vertex $v$. One
way to specify the coloring constraint is the single constraint $x_v
\ne x_w$ for $(v,w)\in E$, the other way is to consider $k$ separate
constraints of ``if $x_u=\alpha$ then $x_v\ne \alpha$''. Following the
single constraint definition, there will be implied commutativity
relations between variables say $x_{u,0}$ and $x_{v,1}$, while in the
second formulation of the problem, there will be no such implied
conditions. It turns out that the latter definition is consistent with
the one defined in~\cite{CMN+07}.

\section{Reductions of BCS problems}

\label{sec:reductions}

Now that we have many interesting examples of $*$-problems, it is
natural to investigate the relations between the structure and
complexity of these different problems. In this section, we propose a
particular type of reduction that is suitable for such discussions.
The main idea is to exploit the local commutativity of $*$-problems in
the design of the reductions. We call such reductions of $*$-problems
the locally commutative reductions or $*$-reductions to emphasize its
key features. The first $*$-reduction that we discuss is in the
following theorem.

\begin{theorem}\label{thm:3coloring}
  \problem*{3-SAT} is polynomial-time (Karp) reducible to
  \problem*{3-COLORING}.
\end{theorem}

\begin{proof}
  Suppose we have an instance of the \problem*{3-SAT} problem,
  $\bigwedge_{j=1}^m C_j$, where each clause is a disjunction of at
  most three literals. The aim is to construct a graph $G$ such that
  the \problem*{3-SAT} instance has an operator assignment if and only
  if the graph has quantum chromatic number less than or equal to $3$.

  \begin{figure}[htbp]
    \centering
    \includegraphics{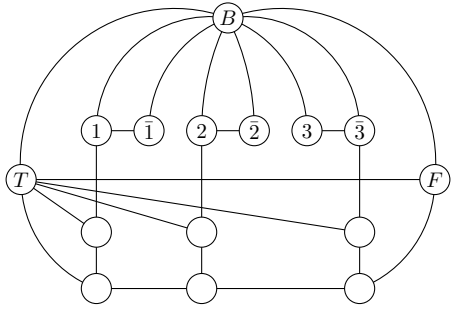}
    \caption{Classical gadget for reduction between \problem{3-SAT}
      and \problem{3-COLORING}}
    \label{fig:3coloring}
  \end{figure}
  The first idea is to check if the reductions for the classical
  problems \problem{3-SAT} and \problem{3-COLORING} will work here. We
  choose one particular such reduction, but it turns out the exact
  choice is not important. In the classical reduction, There will be
  three vertices $F, T, B$ that form a triangle in the graph $G$. They
  will have different colors in any $3$-coloring of the graph, and
  without loss of generality, assume that $F, T, B$ are colored with
  colors $0, 1, 2$ respectively. For each variable $x_j$, create two
  adjacent vertices $j, \bar{j}$ both connected to $B$. We will refer
  to these vertices as the variable vertices. For each clause, say
  $C=x_1 \vee x_2 \vee \neg x_3$, construct a gadget graph as in
  Fig.~\ref{fig:3coloring}. The reason that the reduction works is
  that the gadget is $3$-colorable if and only if the vertices $1, 2,
  \bar{3}$ are not all colored $0$.

  We are already halfway done as it can be shown that if the
  \problem*{3-SAT} instance has an operator assignment, then the graph
  $G$ constructed above indeed has a $3$-coloring operator assignment.
  For each vertex $v$, there will be a tuple of three operators
  assigned, denoted by $(X_{v,0}, X_{v,1}, X_{v,2})$. For vertices $F,
  T, B$, the tuples of operators are $(I, 0, 0)$, $(0, I, 0)$, $(0, 0,
  I)$. For each variable vertex $j$ corresponding to variable $x_j$,
  the tuple is $(I-X_j, X_j, 0)$, and for vertex $\bar{j}$ the tuple
  is $(X_j, I-X_j, 0)$. We argue that the rest of the vertices of each
  gadget in Fig.~\ref{fig:3coloring} can also be assigned a tuple of
  operators such that all the constraints of the \problem*{3-COLORING}
  problem are met. The reason is that, by the local commutativity
  condition, the operators assigned to the variables of each clause in
  the \problem*{3-SAT} problem are pairwise commuting, and can be
  simultaneously diagonalized in some basis. In that basis, the
  problem is essentially classical and the operator assignment can be
  found for each of the remaining vertex in the gadget by the
  correctness of the classical reduction.

  Unfortunately, however, if the graph $G$ is quantum $3$-colorable,
  it may not guarantee the satisfiability of the \problem*{3-SAT}
  instance. For example, any operator assignment for the
  \problem*{3-SAT} problem requires the operators assigned to the
  variables of each clause to be pairwise commuting. However, if we
  extract operators from the operators assigned to the variable
  vertices in a quantum $3$-coloring assignment, commutativity is not
  always promised. The natural way to fix this problem is to come up
  with some other gadget that guarantees the commutativity and modify
  the graph $G$ accordingly.

  The requirements for such gadgets are twofold. First, it will
  guarantee commutativity of operators assigned to two particular
  vertices (commutativity condition). Second, in order to reuse the
  classical reduction, it is required that for any two tuples of
  commuting projection operators summing to $I$ assigned to the two
  vertices, one can always extend the assignment to all the remaining
  vertices in the gadget so that the completed assignment is a quantum
  $3$-coloring assignment (extendibility condition). We call such
  gadgets the commutativity gadgets. By Lemma~\ref{lem:prism}, there
  is such a commutativity gadget for the \problem*{3-COLORING}
  problem, which turns out to be the graph of a triangular prism as in
  Fig.~\ref{fig:prism}.

  Now, for any two non-adjacent vertices $u, v$ in the classical
  gadgets as in Fig.~\ref{fig:3coloring}, attach a commutativity
  gadget and identify vertices $u, v$ with the vertices $a, e$ in the
  gadget respectively. Let us call the resulting graph $G^*$. It is
  easy to see that the proof of \problem*{3-SAT} satisfiability
  implying quantum $3$-colorability of $G^*$ remains almost unchanged.
  We will focus on the converse direction in the following.

  Suppose the graph $G^*$ has a quantum $3$-coloring. The
  commutativity gadget and Lemma~\ref{lem:link} imply that operators
  assigned to the vertices of the classical gadget in the quantum
  $3$-coloring assignment are pairwise commuting. Especially, the
  operators of vertices $F, T, B$ actually commute with all other
  assigned operators in graph $G^*$. It is therefore possible to
  diagonalize the operators for the vertices $F, T, B$ so that each of
  them is a direct sum of six operators corresponding to the six
  different permutation of colors $0, 1, 2$ for these three vertices.
  That is the operators for them can be written as
  \begin{equation*}
    X_{F,\alpha} = \bigoplus_{j=1}^6 f^{(\alpha)}_j I_{\H_j},\;
    X_{T,\alpha} = \bigoplus_{j=1}^6 t^{(\alpha)}_j I_{\H_j},\;
    X_{B,\alpha} = \bigoplus_{j=1}^6 b^{(\alpha)}_j I_{\H_j},
  \end{equation*}
  for $\alpha = 0, 1, 2$, where
  \begin{equation*}
    \begin{split}
      f^{(0)} & = (1, 1, 0, 0, 0, 0),\\
      f^{(1)} & = (0, 0, 1, 1, 0, 0),\\
      f^{(2)} & = (0, 0, 0, 0, 1, 1),\\
      t^{(0)} & = (0, 0, 1, 0, 1, 0),\\
      t^{(1)} & = (1, 0, 0, 0, 0, 1),\\
      t^{(2)} & = (0, 1, 0, 1, 0, 0),\\
      b^{(0)} & = (0, 0, 0, 1, 0, 1),\\
      b^{(1)} & = (0, 1, 0, 0, 1, 0),\\
      b^{(2)} & = (1, 0, 1, 0, 0, 0).
    \end{split}
  \end{equation*}
  As all other operators in the assignment commute with these nine
  operators, they can also be written as a direct sum of operators on
  the six Hilbert spaces. Therefore, the set of operators restricted
  to any of the Hilbert spaces is also an operator assignment. Without
  loss of generality, consider operators restricted to space $\H_1$.
  As the operators in each gadget commute, we can use the classical
  reduction and assign operators $X_{j,1}$ to each variable $x_j$ in
  the \problem*{3-SAT} instance.
\end{proof}

To prove that the commutativity gadget in Lemma~\ref{lem:prism} works,
we need the following two Lemmas. For simplicity, we will use
$v_\alpha$ to represent the operator assigned to vertex $v$ for color
$\alpha$ (previously denoted as $X_{v,\alpha}$), and use $1$ to
represent $I$. Another way of seeing these notions of operators is to
think of them as non-commutative variables and the reasoning of them
in the proof as polynomial identities of these non-commutative
variables.

\begin{lemma}\label{lem:link}
  For any $3$-coloring operator assignment of a graph, the operators
  assigned to adjacent vertices always commute.
\end{lemma}

\begin{proof}
  Let $u, v$ be the two adjacent vertices. The commutativity of
  $u_\alpha$ and $v_\alpha$ is easy to see as their products $u_\alpha
  v_\alpha$ and $v_\alpha u_\alpha$ are both $0$.

  By the symmetry of the problem, it suffices then to show that $u_0$
  and $v_1$ commutes. Consider the following two identities
  \begin{equation*}
    \begin{split}
      u_2 (v_0 + v_1 + v_2 - 1) u_0 - u_0 (v_0 + v_1 + v_2 - 1)
      u_2 & = u_2v_1u_0 - u_0v_1u_2 + [u_0, u_2],\\
      u_0 v_1 (u_0 + u_1 + u_2 - 1) - (u_0 + u_1 + u_2 - 1) v_1 u_0 &
      = u_0v_1u_2 - u_2v_1u_0 + [v_1, u_0].
    \end{split}
  \end{equation*}
  In the computation of the above two identities we have used the
  coloring constraints $u_\alpha v_\alpha = 0$, $v_\alpha u_\alpha =
  0$. Taking the summation of the two identities and noticing that the
  left hand side of them are both $0$ and that $u_0$, $u_2$ commute,
  we have $[v_1, u_0] = 0$.
\end{proof}

We note that the above lemma is a special property of $3$-coloring and
may not hold in the $k$-coloring case for $k > 3$.

\begin{lemma}\label{lem:triangle}
  If vertices $u, v, w$ form a triangle in a graph, their assigned
  operators in a $3$-coloring operator assignment satisfy $u_\alpha +
  v_\alpha + w_\alpha = 1$ for $\alpha = 0, 1, 2$.
\end{lemma}

\begin{proof}
  Lemma~\ref{lem:link}, when applied to the three edges of the
  triangle, shows that any two operators in the quantum assignment to
  the triangle commute. Therefore, all the operators can be
  diagonalized and the problem is essentially classical. The required
  identity then simply corresponds to the fact that one of the
  vertices will be colored with color $\alpha$. It is also possible to
  prove these identities by non-commutative Gr\"{o}bner basis
  computations, which will require more work though.
\end{proof}

\begin{lemma}\label{lem:prism}
  The triangular prism graph as in Fig.~\ref{fig:prism} is a
  commutativity gadget of vertices $a$ and $e$ for
  \problem*{3-COLORING}. By the symmetry of the graph, any
  $3$-coloring operator assignment to the graph is pairwise commuting,
  and for any two tuples of three commuting projection operators
  summing to $I$ assigned to $a$, $e$ respectively, it is always
  possible to extend the assignment to the remaining vertices such
  that the operators form a $3$-coloring operator assignment.
\end{lemma}

\begin{figure}[htbp]
  \centering
  \includegraphics{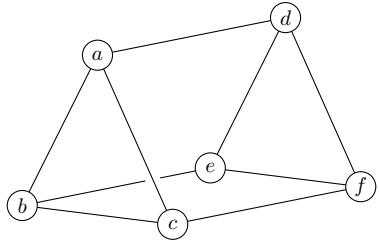}
  \caption{Triangular prism as a commutative gadget for
    \problem*{3-COLORING}}
  \label{fig:prism}
\end{figure}

\begin{proof}
  We will focus on the commutativity part of the proof as the
  extendibility is easy to verify classically.

  First, we show that $a_\alpha, e_\alpha$ commute for any $\alpha =
  0, 1, 2$. Without loss of generality, assume $\alpha = 0$. We have
  \begin{equation*}
    \begin{split}
      c_0 f_0 a_0 - a_0 f_0 c_0 & = (1-a_0-b_0) (1-d_0-e_0) a_0 - a_0
      (1-d_0-e_0) (1-a_0-b_0)\\
      & = [a_0, e_0],
    \end{split}
  \end{equation*}
  where the first equation follows from Lemma~\ref{lem:triangle}, and
  the second is a simplification using coloring constraints such as
  $d_0a_0=0$. Now by $c_0f_0=f_0c_0 = 0$, the condition $[a_0, e_0] =
  0$ follows.

  Next, we prove that $[a_\alpha, e_\beta] = 0$ where $\alpha$ is
  different from $\beta$. Without loss of generality, assume
  $\alpha=0$, $\beta=1$. Compute the following quantity
  \begin{equation*}
    \begin{split}
      e_2 d_2 a_0 - a_0 d_2 e_2 & = (1-e_0-e_1) (1-d_0-d_1) a_0 - a_0
      (1-d_0-d_1) (1-e_0-e_1)\\
      & = [a_0, d_1] + [a_0, e_0] + [a_0, e_1] + e_0d_1a_0 - a_0d_1e_0\\
      & = [a_0,e_1],
    \end{split}
  \end{equation*}
  where the first equation follows from the coloring constraints
  $e_0+e_1+e_2=1$ and $d_0+d_1+d_2=1$, the second equation is a
  simplification using coloring constraints such as $d_0a_0 = 0$ and
  $e_1d_1=0$, and the last equation uses Lemma~\ref{lem:link} and the
  first part of this lemma. The lemma now follows by the fact that
  $e_2d_2=d_2e_2=0$.

  We mention that the proof of this lemma was first obtained by
  computing the non-commutative Gr\"{o}bner basis of the ideal
  generated by the constraints and check that the commutators are in
  the ideal. A verification program {\tt prism.g} using the {\it
    GBNP\/} package~\cite{CK10} in the {\it GAP\/} system~\cite{GAP4}
  can be found in the
  \href{http://arxiv.org/e-print/1310.3794}{gzipped tar (.tar.gz)
    file} of the arXiv version of this paper.
\end{proof}

The method of combining a classical reduction with a suitable
commutativity gadget, as illustrated in the proof of
Theorem~\ref{thm:3coloring}, is applicable to finding $*$-reductions
of other problems. Another such example worth mentioning is the
reduction of \problem*{3-SAT} to \problem*{1-in-3-SAT}.

The problem \problem{1-in-3-SAT} is a variant of \problem{3-SAT} that
requires one and only one of three literals to be true in each clause.
It is shown to be \class{NP}-complete by Schaefer~\cite{Sch78}. In
fact, the monotone version of it (where no negative literals occur)
remains \class{NP}-complete. The constraint that exactly one of the
three variables is true can be described by a linear equation of the
form $x_{1} + x_{2} + x_{3} = 1$, where the addition is over real
numbers. Therefore, \problem*{1-in-3-SAT} is a special
\problem*{KOCHEN-SPECKER} problem, and the following theorem also
implies that \problem*{3-SAT} is $*$-reducible to
\problem*{KOCHEN-SPECKER}.

\begin{theorem}
  \problem*{3-SAT} is polynomial-time (Karp) reducible to
  \problem*{1-in-3-SAT}.
\end{theorem}

\begin{proof}
  As in the proof of Theorem~\ref{thm:3coloring}, we start with the
  classical reduction given in~\cite{Sch78}. Let $R$ be a polynomial
  such that $R(x,y,z) = 1$ if and only if exactly one of $x, y, z$ is
  $1$. For example, one can take $R(x,y,z) = x+y+z$ as discussed
  above. For each clause $C=x\vee y\vee z$, the classical reduction
  uses six new variables $u_1, u_2, \ldots, u_6$, and introduces five
  clauses
  \begin{equation}
    \label{eq:1in3}
    R(x, u_1, u_4), R(y, u_2, u_4), R(u_1, u_2, u_3),
    R(u_4, u_5, u_6), R(z, u_5, 0).
  \end{equation}
  The claim is that these five clauses can all be satisfied if and
  only if clause $C$ is satisfied. (Negative literals such as $\neg x$
  can be dealt with by introducing a variable $x'$ and a constraint $x
  + x' = 1$.)

  By the computation of the non-commutative Gr\"{o}bner basis, one can
  verify that the classical gadget already implies the commutativity
  of $x$ and $y$, but cannot guarantee the commutativity for $x, z$ or
  $y, z$. The commutativity gadget can therefore be extracted from the
  classical reduction by considering the first three of the five
  clauses in Eq.~\eqref{eq:1in3}. We summarize this observation in
  Lemma~\ref{lem:1in3}.
\end{proof}

\begin{lemma}\label{lem:1in3}
  The binary constraint system
  \begin{subequations}
    \begin{align}
      x\phantom{_1} + u_1 + u_4 & = 1,\\
      y\phantom{_1} + u_2 + u_4 & = 1,\\
      u_1 + u_2 + u_3 & = 1,
    \end{align}
  \end{subequations}
  forms a commutativity gadget of $x, y$ for \problem*{1-in-3-SAT}.
\end{lemma}

\begin{proof}
  The extendibility condition is easy to verify and we focus on
  proving the commutativity condition. Compute the commutators
  \begin{equation*}
    \begin{split}
      [x+u_1+u_4-1, -x+u_1+u_3] & = [x, u_3] + [u_4,
      u_3],\\
      [y+u_2+u_4-1, -x] & = [x, y] + [x, u_2],\\
      [u_1+u_2+u_3-1, x+u_4] & = [u_2, x] + [u_3, x] + [u_3, u_4].
    \end{split}
  \end{equation*}
  In the computation, we have used the local commutativity conditions
  such as $[x,u_1] = 0$. Noticing that the right hand side of each
  equation is $0$, we have $[x, y]=0$ by taking the sum of all three
  equations. It is worth mentioning that one can prove the
  commutativity of $x, y$ even without using local commutativity
  conditions, but the proof would be much harder.
\end{proof}

Sometimes, the commutativity gadget is much easier to construct. For
example, if we want to reduce some problem to \problem*{3-SAT}, the
commutativity gadget of $x, y$ is simply a new clause $x\vee y\vee z$
where $z$ is a new variable not occurring in the classical reduction.
Another such example is the \problem{NAE-SAT} constraint. In these
examples, commutativity follows from the local commutativity
condition, and extendibility is guaranteed by the nature of these
constraints.

\begin{theorem}
  \problem*{$k$-SAT} is polynomial-time (Karp) reducible to
  \problem*{3-SAT}.
\end{theorem}

\begin{proof}
  Following the classical reduction, transform each clause of the form
  $C=\bigvee_{j=1}^k x_j$ to a conjunction of $k-2$ clauses
  \begin{equation*}
    (x_1\vee x_2 \vee y_1) \wedge (\neg y_1 \vee x_3 \vee y_2) \wedge
    \cdots \wedge (\neg y_{k-4} \vee x_{k-2} \vee y_{k-3}) \wedge (\neg
    y_{k-3} \vee x_{k-1} \vee x_k),
  \end{equation*}
  where $y_1, y_2, \ldots, y_{k-3}$ are new variables. To recover
  commutativity lost from the reduction, add a new clause $x\vee y\vee
  z$ for each pair of variables $x, y$ from $x_1, x_2, \ldots, x_k,
  y_1, y_2, \ldots, y_{k-3}$ not occurring in the same clause.
\end{proof}

As \problem*{$k$-COLORING} is a special case of \problem*{$k$-SAT},
\problem*{$k$-COLORING} is polynomial-time reducible to
\problem*{3-COLORING}, although the graph for the latter problem may
have larger size. More generally, the problem of determining the
quantum chromatic number of a graph can be reduced to a bunch of
\problem*{3-COLORING} problems.

The commutativity gadget can also be used to prove \class{NP}-hardness
of many $*$-problems. For example, we have

\begin{theorem}
  \problem*{3-SAT} is \class{NP}-hard.
\end{theorem}

\begin{proof}
  We prove the result by reducing \problem{3-SAT} to \problem*{3-SAT}.
  Let the \problem{3-SAT} instance of $n$ variables $x_1, x_2, \ldots,
  x_n$ be $\bigwedge_{j=1}^m C_j$. For each pair of variables $x_i,
  x_j$ that do not occur in the same clause, introduce a new clause
  $x_i \vee x_j \vee y$, the commutativity gadget, where $y$ is a new
  variable. It is then easy to see that the resulting instance has
  quantum satisfying assignment if and only if the original instance
  has a classical satisfying assignment.
\end{proof}

There are several remarks on the locally commutative reductions in
general. First, they usually also preserve the correctness of the
classical base reductions. Therefore, if we start with a magic BCS,
the resulting instance is also a magic BCS. Since this is true even
without the commutativity gadget, the classical reduction of
\problem{3-SAT} to \problem{3-COLORING} can be used to construct graph
that has quantum $3$-coloring but no classical $3$-coloring from the
magic square game or any other magic BCS games. Second, it is also
usually true that the dimension of the operator assignment is
preserved in the reduction. Therefore, in order to give upper bound on
the entanglement for BCS games, it suffices to work with special types
of constraints in any of the above discussed examples.

In the end of this section, we mention that the $*$-versions of
\problem{2-SAT} and \problem{HORN-SAT} remain in \class{P}.

\begin{theorem}\label{thm:2sat}
  \problem*{2-SAT} is in \class{P}. In fact, it is the same problem as
  its classical counterpart \problem{2-SAT}.
\end{theorem}

\begin{proof}
  Any \yes-instance of \problem{2-SAT} is also a \yes-instance of
  \problem*{2-SAT}. So it suffice to prove that any \no-instance of
  \problem{2-SAT} is \no-instance of \problem*{2-SAT}. For any
  \problem{2-SAT} instance, its implication graph~\cite{APT79} is a
  directed graph with all variables and their negations as its
  vertices. For each clause $x\vee y$, add two edges, one from $\neg x
  $ to $y$, the other from $\neg y$ to $x$. Consider the implication
  graph of a \no-instance. By the result of~\cite{APT79}, there will
  be a strongly connected component (subgraphs such that there is a
  path form each vertex in the component to every other vertex in the
  component) that contains a literal and its negation.

  Suppose, on the contrary, that there is a quantum assignment,
  assigning operator $X$ to variable $x$. To each vertex of graph
  labeled by a variable $x$, assign the $1$-eigenspace of $X$, and to
  vertex labeled by $\neg x$, assign the $0$-eigenspace of $X$. The
  constraints of \problem*{2-SAT} then translate to the following: for
  each directed edge $(x,y)$ in the implication graph, the subspace
  assigned to vertex $x$ is contained in the subspace assigned to $y$.
  Let $x$ and $\neg x$ be the pair of vertices in a strongly connected
  component~\cite{APT79}. We will then have that the $1$-eigenspace of
  $X$ is contained in its $0$-eigenspace and vice versa, which is a
  contradiction.
\end{proof}

The above theorem is a slight generalization of the fact that a graph
is quantum $2$-colorable if and only if it is classical $2$-colorable.
Also, it makes clear that it is important to have constraints like
$x_{v,0}+x_{v,1}+x_{v,2} = 1$ in, for example, the quantum
$3$-coloring problem, as all other constraints are indeed
\problem{2-SAT} constraints.

The \problem{HORN-SAT} problems are satisfiability problems where each
constraint is a Horn clause (a disjunction with at most one positive
literal). We have the following theorem for its $*$-version.

\begin{theorem}\label{thm:hornsat}
  \problem*{HORN-SAT} is in \class{P}. In fact, it is the same problem
  as its classical counterpart \problem{HORN-SAT}.
\end{theorem}

\begin{proof}
  As in the proof of Theorem~\ref{thm:2sat}, it suffices to prove that
  any \no-instance of \problem{HORN-SAT} is also a \no-instance of
  \problem*{HORN-SAT}. Following the convection in~\cite{DG84}, for a
  \problem{HORN-SAT} instance $H = \bigwedge_{j=1}^m C_j$ of $n$
  variables, define a labeled directed graph $G_H$ with $n+2$ vertices
  (a vertex for each variable, a vertex for \true, and a vertex for
  \false). For $j=1, 2, \ldots, m$, construct the edges of $G_H$ as
  follows: (a) If $C_j$ is a positive literal $x$, add an edge from
  \true to $x$ labeled $j$, (b) if $C_j$ is of the form $\neg x_1 \vee
  \cdots \vee \neg x_k$, add $k$ edges from $x_1, \ldots, x_k$ to
  \false labeled $j$, (c) if $C_j$ is of the form $\neg x_1 \vee
  \cdots \vee \neg x_k \vee y$, add $k$ edges from $x_1, \ldots, x_k$
  to $y$ labeled $j$. In the graph $G_H$, there is a {\it pebbling\/}
  of a vertex $y$ from a set $X$ of vertices if either $y$ belongs to
  $X$ or, for some label $j$, there are pebblings of $x_1, \ldots x_k$
  from $X$ and $x_1, \ldots, x_k$ are the sources of all incoming
  edges to $y$ labeled $j$. Theorem 3 and its corollary of~\cite{DG84}
  then states that $H$ is unsatisfiable if and only if there is a
  pebbling of \false from $\{ \true \}$ in $G_H$.

  Let $H$ be a \no-instance of \problem{HORN-SAT} and assume on the
  contrary that there is a quantum satisfying assignment. For a
  variable $x$ in $H$, let $X$ be the operator assigned to it. We
  assign a subspace to each vertex of the graph $G_H$ in the following
  way. For vertex \false, assign the $0$-dimensional subspace $\{0\}$,
  for vertex \true, assign the whole Hilbert space $\H$ of the
  assignment, and for each vertex $x$, assign the $1$-eigenspace of
  $X$. The constraints of the \problem*{HORN-SAT} problem then
  translate to that, for each $1\le j \le m$, the intersection of the
  subspaces assigned to the sources of edges labeled by $j$ is
  contained in the subspace assigned to the target of these edges. By
  induction, it is easy to show that if there is a pebbling of $y$
  from $X$, then the intersection of the subspaces assigned to the
  vertices in $X$ is contained in that assigned to $y$. We thus have a
  contradiction that $\H$ is contained in $\{0\}$ by the
  characterization of~\cite{DG84}.
\end{proof}

From the above discussions, the $*$-versions of the few
polynomial-time solvable satisfiability problems singled out in
Schaefer's dichotomy classification theorem~\cite{Sch78} are the same
as their classical counterparts except the affine case. This makes the
determination of the complexity of the parity BCS problems
($*$-version of the affine case) particularly interesting. Currently,
the problem is not even known to be decidable. It is also the only
case where the $*$-version is indeed different from the classical
problem, and the difference, if measured by the amount of entanglement
involved, is big as shown in the next section.

\section{A Parity BCS Game that Requires a Large Amount of
  Entanglement}

\label{sec:lower}

We have seen in Section~\ref{sec:reductions} that the commutativity
gadgets played a crucial role in constructing the $*$-reductions. In
this section, we will employ a variant of it, called the
anti-commutativity gadget, to construct a BCS game that requires a
large amount of entanglement in its prefect strategies.

It turns out that the magic square game is a naturally born
anti-commutativity gadget. The following two simple observations about
magic square system correspond to the anti-commutativity condition and
the extendibility condition of the gadget respectively.

\begin{lemma}\label{lem:anticommute}
  If $X_1, X_2, \ldots, X_9$ form an operator assignment of the magic
  square constraint system, then the following anti-commutativity
  relation holds:
  \begin{equation*}
    X_2X_4=-X_4X_2.
  \end{equation*}
  In fact, for any two operators not in the same row or column of the
  square are always anti-commuting.
\end{lemma}

\begin{proof}
  From $X_1X_2X_3 = \1$ and $X_3^2 = \1$, we have that $X_3 = X_1X_2$.
  Similarly, $X_6 = X_4X_5$, $X_7 = X_1X_4$, $X_8 = X_2X_5$ and
  $X_7X_8 = X_9 = - X_3X_6$. It then follows by substitution that
  $(X_1X_4)(X_2X_5) = -(X_1X_2)(X_4X_5)$, which implies that $X_2X_4 =
  -X_4X_2$. The other anti-commutativity relations can be shown
  similarly.
\end{proof}

\begin{lemma}\label{lem:extension}
  If $A, B$ are any two self-adjoint bounded linear operators on
  Hilbert space $\H$ that satisfy $AB = -BA$, and $A^2 = B^2 = \1$,
  then there exists an operator assignment $Y_1, Y_2, \ldots, Y_9$ on
  $\H_2\otimes \H$ for the magic square game such that $Y_2 = \1
  \otimes A$, $Y_4 = \1 \otimes B$, where $\H_2$ is the two
  dimensional complex Hilbert space.
\end{lemma}

\begin{proof}
  Choose $Y_j$'s as in Fig.~\ref{fig:extension} where $C = iAB$ and
  $I,X,Y,Z$ are the identity and Pauli matrices. It is easy to verify
  that they indeed form an operator assignment for the magic square
  BCS. For example, in the last column, we have
  \begin{equation*}
    \begin{split}
      & (X \otimes A) (Z \otimes B) (Y \otimes C) = iXZY \otimes ABAB
      = - \1.
    \end{split}
  \end{equation*}
  \begin{figure}[htbp]
    \centering
    \includegraphics{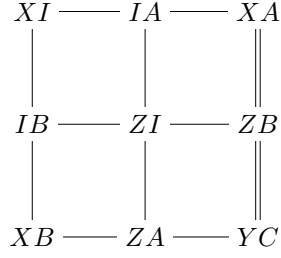}
    \caption{Extension of an anti-commuting pair to a magic square}
    \label{fig:extension}
  \end{figure}
\end{proof}

We now construct a BCS game as follows. The basic idea is to glue a
bunch of magic squares and use the anti-commutativity relations
implicit in the magic square to form the defining relations of the
Clifford algebra. Let $G$ be a complete graph with vertex set
$V=\{1,2,\ldots,N\}$. To each vertex $j\in V$, we assign a variable
$x_j$, and with each edge $e=(j,k)$ for $j<k$, we associate seven
variables $y_1^{(e)}, y_2^{(e)}, \ldots, y_7^{(e)}$. They form the
complete set of variables of the BCS. The total number of variables is
therefore of order $N^2$. For each edge $e=(j, k)$, we form a magic
square using $x_j, x_k,$ and $y_1^{(e)}, y_2^{(e)}, \ldots, y_7^{(e)}$
as in Fig.~\ref{fig:edge}. That is, we have $y_1^{(e)} x_j y_2^{(e)} =
1$, $x_k y_3^{(e)} y_4^{(e)} = 1$, etc., in the constraints for each
edge $e=(j, k)$. The number of constraints is therefore also of order
$N^2$. Let us call the BCS game thus obtained the Clifford BCS game of
rank $N$.

\begin{figure}[htbp]
  \centering
  \includegraphics{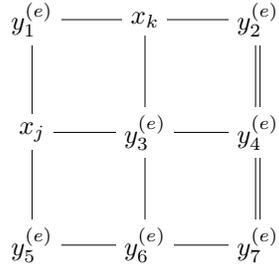}
  \caption{Magic square for edge $e$}
  \label{fig:edge}
\end{figure}

\begin{theorem}
  The Clifford BCS game of rank $N$ has a perfect strategy and
  requires at least $\lfloor \frac{N}{2} \rfloor$ shared EPR pairs for
  Alice and Bob to play perfectly.
\end{theorem}

\begin{proof}
  We first prove that the Clifford BCS game has a perfect strategy.
  For this purpose, we first choose operators $X_j$ for each vertex
  variable $x_j$, $j\in V$, such that $X_j^2=\1$ and $X_jX_k=-X_kX_j$.
  By the well-known representation theory of the Clifford
  algebra~\cite[Subsection 17.3]{Kir76}, these operators exist on a
  $2^{\lfloor \frac{N}{2} \rfloor}$ dimensional Hilbert space $\H$. In
  fact, one can even choose the operators to be tensor product of
  Pauli operators. Then a repeated application of
  Lemma~\ref{lem:extension} to the magic squares on all edges gives an
  operator assignment on $\H_2\otimes \H$. Therefore, there exists a
  perfect strategy for the Clifford BCS game.

  To prove the lower bound on entanglement, it suffices to observe
  that for any operator assignment of the Clifford game, the operator
  assignment to the vertex variables $x_j$'s forms a representation of
  the Clifford algebra of rank $N$ by Lemma~\ref{lem:anticommute}. It
  follows from representation theory of the Clifford
  algebra~\cite[Subsection 17.3]{Kir76} that the Hilbert space on
  which these operators act is at least $2^{\lfloor \frac{N}{2}
    \rfloor}$ dimensional. This already implies that any perfect
  strategy requires at least $\lfloor \frac{N}{2} \rfloor$ shared EPR
  pairs.
\end{proof}

We remark that the above theorem establishes a bound of
$\Omega(\sqrt{n})$ shared EPR pairs in term of the number of variables
$n = \Theta(N^2)$. One can reuse many of the edge variables in the
construction, but the order will always be $N^2$. The above
construction is straightforward and it is conceivable that more
contrived designs can give better entanglement bounds. A possible
approach is to consider BCS game derived not from a complete graph,
but from some arbitrary graph instead. More concretely, let $G=(V,E)$
be a graph. For each $j\in V$ assign a self-adjoint operator $X_j$
such that $X_j^2 = \1$, and for each $(j,k) \in E$, require that
$X_jX_k=-X_kX_j$. One can then turn these constraints into a BCS in a
similar way. The difficulty, however, is that we do not know a lower
bound on the dimension for the operators $X_j$'s, which satisfy only a
partial set of anti-commuting relations. An easy upper bound on this
dimension is $2^{O(\chi(G))}$ where $\chi(G)$ is the chromatic number
of graph $G$. Hence, such an improvement could only be possible for
certain sparse graphs with high chromatic numbers.

In the above discussion of parity BCS games, the commutativity and
anti-commutativity conditions played an important role. However, we
mention that parity BCS games designed based on a {\it complete\/} set
of pairwise commuting or anti-commuting variables will need no more
than $n$ EPR pairs for perfect strategies to exist. The reason is
that, in this case, the group generated by the variables $x_1, x_2,
\ldots, x_n$ will be a finite group whose elements are all in the form
$\pm x_1^{e_1} x_2^{e_2} \ldots x_n^{e_n}$ for $e_j\in \{0, 1\}$. By
considering the regular representation of the group, we know that
there is operator assignment of dimension at most $2^n$. It turns out
that all known magic parity BCS games have a Pauli solution. For all
such games, the upper bound of $n$ EPR pairs applies as the existence
of Pauli solutions indicates that certain complete set of
commutativity and anti-commutativity relations is consistent with the
constraint system.

Finally, we explain how one can reduce the number of occurrences of
variables in BCS games. In the Clifford BCS game defined above, there
are variables (those correspond to the vertices for example) that
occur in $\Omega(N)$ clauses. However, it is easy to lower this number
all the way down to three in all BCS games by introducing several
auxiliary variables and constraints. Consider for example a binary
tree. Assign a variable $z_j$ to each vertex of the tree, and let each
edge of the tree to represent an even parity constraint. That is, if
$z_j$ and $z_k$ are two variables of neighboring vertices, then
require that $z_j \oplus z_k = 0$. This way, each vertex occurs in at
most three clauses, and each leaf occurs in only one clause. It is
easy to see that the variables in the tree must all be equal.
Therefore, we can construct such a binary tree of suitable size for
each variable in the original BCS game that occurs in more than three
clauses and replace the occurrences of this variable with the leaf
variables of the tree.

\section*{Acknowledgments}

The author acknowledges helpful discussions with Richard Cleve and
John Watrous on related problems. This work is supported by NSERC and
ARO.

\bibliographystyle{acm}

\bibliography{bcs-reduction}

%% Start-Of-Trailer

\end{document}